\documentclass[letterpaper, 10 pt, conference]{ieeeconf}  

\IEEEoverridecommandlockouts                              

\usepackage{graphicx} 
\usepackage{amsmath} 
\usepackage{amssymb,latexsym,amsfonts}
\usepackage{graphicx}
\usepackage{diagrams}
\usepackage{epsfig}
\usepackage{graphicx}
\usepackage{color}
\usepackage{mathtools}
\usepackage{relsize}
\usepackage{tikz}
\usetikzlibrary{automata}
\usetikzlibrary{shapes}
\usepackage{amssymb,latexsym,amsfonts,amsmath}
\usepackage{graphicx}
\usepackage[normalem]{ulem}
\usepackage{dsfont}
\usepackage{refcount}
\usepackage{leftidx}
\usepackage{comment}
\usepackage{psfrag}
\usepackage{epstopdf}

\usepackage{accents}
\newcommand\ubar[1]{%
  \underaccent{\bar}{#1}}
\newtheorem{mydefinition}{Definition}
\newtheorem{mytheorem}{Theorem}
\newtheorem{myproposition}{Proposition}
\newtheorem{mycorollary}{Corollary}
\newtheorem{mylemma}{Lemma}
\newtheorem{myproblem}{Problem}
\newtheorem{myremark}{Remark}

\usepackage{mathtools}
\newlength{\arrow}
\usepackage[linesnumbered,algoruled]{algorithm2e}

\usepackage{color}

\textfloatsep=\floatsep
\title{\LARGE \bf
Technical Report:
Timing Abstraction of Perturbed LTI systems with $\mathcal{L}_2$-based Event-Triggering Mechanism}

\author{Arman Sharifi Kolarijani, Manuel Mazo Jr. and Tam\'{a}s Keviczky \thanks{The authors are with the Delft Center for Systems and Control, Delft University of Technology, The Netherlands.\newline {\tt\footnotesize \{a.sharifikolarijani,m.mazo,t.keviczky\}@tudelft.nl}}}

\begin{document}

\maketitle
\thispagestyle{empty}
\pagestyle{empty}

\begin{abstract}

In networked control systems, the advent of event-triggering strategies in the sampling process has resulted in the usage reduction of network capacities, such as communication bandwidth. However, the aperiodic nature of sampling periods generated by event-triggering strategies has hindered the schedulability of such networks. In this study, we propose a framework to construct a timed safety automaton that captures the sampling behavior of perturbed LTI systems with an $\mathcal{L}_2$-based triggering mechanisms proposed in the Literature. In this framework, the state-space is partitioned into a finite number of convex polyhedral cones, each cone representing a discrete mode in the abstracted automaton. Adopting techniques from stability analysis of retarded systems accompanied with a polytopic embedding of time, LMI conditions to characterize the sampling interval associated with each region are derived. Then, using reachability analysis, the transitions in the abstracted automaton are derived.   
\end{abstract}

\section{Introduction}
Wireless networked controlled systems (WNCS's) represent a class of spatially distributed control systems for which the feedback loops are closed via shared communication components possessing limited bandwidth. Several advantages of WNCS’s, such as the ease of maintenance and the flexibility of implementation, make them attractive to industrial environments. Meanwhile, WNCS's are burdened with characteristics, such as limited battery life and communication bandwidth. Under these circumstances, the resource over-utilization caused by (traditional) periodic implementations, the so-called time-driven control (TDC), makes such implementations less appealing for WNCS's.

To address the aforementioned issues, control researchers have proposed \emph{event-driven control} (EDC) strategies that are aperiodic, such as \emph{event-triggered control} (ETC) \cite{Tabuada} and \emph{self-triggered control} (STC) \cite{Velasco}. In EDC strategies, the core idea relies on the fact that the dynamics of the control system during the inter-sample interval determine the next sampling instant to attenuate the usage of resources, particularly the communication bandwidth. In these strategies, control task executions only happen when a pre-specified condition is violated. Such condition is called the \emph{triggering mechanism} (TM). It is derived based on stability and/or performance of the closed-loop system. On the other hand, the \emph{schedulability} of ETC strategies, due to their aperiodic nature, is more arduous compared to TDC strategies. In fact, in TDC strategies, the control and scheduler designs are naturally decoupled via the (pre-defined) fixed sampling period. This phenomenon is called the \emph{separation-of-concerns} in the real-time systems community \cite{dijkstra1982role}. It is worth mentioning that ETC strategies are almost always equipped with a \emph{minimum inter-execution time} (MIET) to prevent the occurrence of \emph{Zeno} behavior in the sampling process. This quantity can be technically used in the synthesis of task scheduling. However, it is a conservative approximation of the lower bound on all the possible generated sampling periods. Thus, such synthesis does not make use of the beneficiary characteristics of ETC strategies in an efficient manner. To address this shortcoming, researchers have proposed another class of approaches, the so-called \emph{co-design} approaches. In this class, the problem of controller and scheduler synthesis for real-time systems is tackled in a unified framework, see e.g.\ feedback modification to task attributes \cite{Buttazzo1998,Caccamo2000}, \cite{Lu2002,Cervin2004}, anytime controllers \cite{Bhattacharya2004,Fontanelli2008_short}, and event-based control and scheduling \cite{AlAreqi2013,AlAreqi2014}. Recently, alternative to the unified frameworks mentioned above, \cite{Arman2015,ArmanTCoN} have proposed a decoupling framework to capture the sampling behavior of LTI systems with ISS-based TM's using timed safety automata (TSA's). 

Generally speaking, TSA is a simplified version of timed automaton (TA) \cite{alur1996timing,henzinger1994symbolic}. It is a powerful tool to model the timing behavior of real-time systems for scheduling purposes since its reachability analysis is decidable \cite{alur1994theory,larsen1997uppaal}. In this study, following the same path as in \cite{Arman2015,ArmanTCoN}, we propose a framework to capture the sampling behavior of perturbed LTI systems with the $\mathcal{L}_2$-based TM proposed by \cite{WangLem2009}. We show that the derived TSA $\varepsilon$-approximately simulates the sampling behavior of the $\mathcal{L}_2$-based ETC system. It is evident that such characterizations can be analyzed independently for scheduling purposes, thus providing a scalable and versatile event-triggered WNCS design procedure.

\section{Preliminaries}
$\mathbb{R}^n$ denotes the $n$-dimensional Euclidean space, $\mathbb{R}^+$ denotes the positive reals. $\mathbb{N}_0$ is the set of nonnegative integers, and $\mathbb{IR}^+$ is the set of all closed intervals $[a,b]$ such that $a,b \in \mathbb{R}^+$ and $a \leq b$. For any set $S$, $2^S$ denotes the set of all subsets of $S$, i.e.\ the power set of $S$. $\mathcal{S}_{m \times n}$ and $\mathcal{S}_n$ are the set of all $m \times n$ real-valued matrices and the set of all $n \times n$ real-valued symmetric matrices, respectively. For a matrix $M$, $M \preceq 0$ (or $M \succeq 0$) means $M$ is a negative (or positive) semidefinite matrix and $M\prec 0$ ($M \succ 0$) indicates $M$ is a negative (positive) definite matrix. $\mathcal{S}_n^+$ is the cone of all $n\times n$ symmetric positive definite matrices. $\lfloor x \rfloor$ indicates the largest integer not greater than $x \in \mathbb{R}$. $|y|$ and $\|M \|$ denote the Euclidean norm of a vector $y \in \mathbb{R}^n$ and the Frobenius norm of a matrix $M \in \mathcal{S}_{m \times n}$, respectively. For a matrix $M \in \mathcal{S}_n$, $\lambda(M)$ and $\lambda_{\max}(M)$ denote the set of eigenvalues and the largest eigenvalue of $M$. Consider two sets $X,Y\subseteq\mathbb{R}^n$, their Minkowski sum is given by $X\oplus Y:= \{x+y| x\in X\mbox{~and~}y \in Y\}$. We state the following known results that will be used in Subsection~\ref{sec::control_pre}.

\begin{mylemma}{(\cite{GuKhCh2003})}
\label{lem_1}
For any real matrices $E$, $G$ and real symmetric positive definite matrix $P$, with compatible dimensions,
\begin{equation*}
EG+G^TE^T \preceq EPE^T+G^TP^{-1}G.
\end{equation*}
\end{mylemma}

\begin{mylemma}{(\cite{Loan1977})}
\label{lem_2}
For all $A \in \mathcal{S}_{n\times n}$, if $\mu(A)= \max \{ \mu \in \mathbb{R}| \; \mu\in \lambda\left(\frac{A^T +A}{2}\right) \}$, then, $|e^{At}| \leq e^{\mu(A)t}$.
\end{mylemma}

\begin{myproposition}{(Jensen Inequality \cite{GuKhCh2003})}
\label{prop_1} For any matrix $M \in \mathcal{S}_m^+$ with constant entries, scalar $\gamma > 0$, vector function $\omega: [0,\gamma] \rightarrow \mathbb{R}^m$ such that the integrations concerned are well defined, then:
\begin{equation*}
\gamma \int_0^{\gamma} \omega^T(\beta)M\omega(\beta)d\beta \geq \left(\int_0^{\gamma}\omega(\beta)d\beta\right)^T M \int_0^{\gamma}\omega(\beta)d\beta.
\end{equation*}
\end{myproposition}

\subsection{$\mathcal{L}_2$-Based ETC System:}
\label{sec::control_pre}
In this subsection, an overview of the ETC strategy proposed by \cite{WangLem2009} along with a new result (see Theorem \ref{theor_1}) are presented. Consider a sampled-data system that is given by:
\begin{equation}
\label{eq_1}
\begin{array}{l}
\dot{\xi}(t)=A \xi (t)+B \nu (t)+E \omega (t), \forall t \in [0, \tau(x)),\\
\xi(0)=x,
\end{array}
\end{equation}
where $ \xi(t) \in \mathbb{R}^n$, $\nu(t) \in \mathbb{R}^m$, $\omega(t) \in \mathbb{R}^p$, $\tau(x)$ denotes the sampling period associated with $\xi(0)$, and $A$, $B$, and $E$ have compatible dimensions. The control law is implemented in a sample-and-hold manner as follows:
\begin{equation}
\label{eq_2}
\nu(t)=-Kx.
\end{equation}

Furthermore, assume that the disturbance $\omega$ is a vanishing type disturbance \cite{WangLem2009}, i.e.,
\begin{equation}
\label{eq_3}
\exists W \geq0 \mbox{~such that~} |w(t)|^2 \leq W |x|^2, \forall t \in [0, \tau(x)).
\end{equation}

Denote by $\epsilon$, the error signal endured by the system (\ref{eq_1})-(\ref{eq_2}), $\epsilon (t) = x - \xi_x (t)$ where $\xi_x(t)$ is the solution of (\ref{eq_1}). Reformulating (\ref{eq_1}), the evolution of state and error signals can be rewritten in a compact form as follows:
\begin{equation}
\label{eq_5}
\xi_x (t) = \Lambda (t) x+\Omega (t)
\end{equation}
and 
\begin{equation}
\label{eq_7}
\epsilon (t)= \left( I- \Lambda(t)\right) x - \Omega(t)
\end{equation}
where
\begin{equation}
\label{eq_6}
\left\{
\begin{array}{l}
\Lambda (t) = I+\int_0^t e^{As}ds(A-BK),\\
\Omega (t) = \int_0^t e^{A(t-s)}E \omega(s) ds.
\end{array} \right.
\end{equation}

Assume that there exists a quadratic Lyapunov function $V(\xi)=\xi^TP\xi$ such that $P$ is the solution to the Algebraic Riccati Equation (ARE) given by:
\begin{equation}
\label{are_1}
PA+A^TP-Q+R=0
\end{equation}
where
\begin{equation}
\label{are_2}
Q=PBB^TP,\quad R=\frac{1}{\gamma^2}PEE^TP,\quad \gamma>0.
\end{equation}

The existence of $V$ guarantees that the system (\ref{eq_1}) with the full-state feedback $\nu(t)=-K\xi(t)=-B^TP\xi(t)$ is finite-gain $\mathcal{L}_2$ stable from $\omega$ to $(x^T,u^T)$ with an induced gain less than $\gamma$ \cite{WangLem2009}. Then, the state-dependent TM, proposed by \cite{WangLem2009}, is given by:
\begin{equation}
\label{eq_8}
\tau(x):=\inf\{t>0|\;\epsilon^T(t)M\epsilon(t) \geq x^T N x\}
\end{equation}
where
\begin{equation}
\label{eq_9}
\begin{array}{lll}
M=(1-\beta^2)I+PBB^TP,\\
N=\frac{1}{2}(1-\beta^2)I+PBB^TP,
\end{array}
\end{equation}
and $\beta >0$ is a user-defined scalar related to the TM (\ref{eq_8}).

\begin{mytheorem}
\label{theor_1}
Consider the system (\ref{eq_1})-(\ref{eq_2}) with the triggering mechanism (\ref{eq_8}). Assume there exist a scalar $\mu$ and a symmetric matrix $\Psi$ such that
\begin{equation}
\label{eq_17_1_a}
\mu\geq0,\quad \Psi\succ0,\quad M+\Psi\preceq \mu I,
\end{equation}
\begin{equation}
\label{eq_17_1_b}
\Phi(t) \succeq 0,
\end{equation}
where 
\begin{equation}
\label{eq_17_123}
\Phi(t) = \left[
\begin{array}{lr}
\Phi_1(t) & \Phi_2(t)\\
\Phi_3(t) & \Phi_4(t)
\end{array}
\right],
\end{equation}
\begin{equation*}
\begin{array}{l}
\Phi_1(t)=(\Lambda(t)-I)^T M (\Lambda(t)-I)\\
\quad\quad\;\;+t W \mu\lambda_{\max}(E^TE)d_A(t)I- N,\\
\Phi_2(t)=\Phi_3^T(t)=(\Lambda(t)-I)^T M^T,\\
\Phi_4(t)=- \Psi,
\end{array}
\end{equation*}
are satisfied. Then, the sampling period $\tau(x)$ generated by (\ref{eq_8}) is lower bounded by:
\begin{equation}
\label{eq_17_4}
\tau'(x):=\inf\{t>0|\;\Phi(t)\succeq0 \}.
\end{equation}
\end{mytheorem}

\begin{proof}
Substitute (\ref{eq_7}) into (\ref{eq_8}). Then, the TM (\ref{eq_8}) can be rewritten by:
\begin{equation}
\label{eq_17}
\tau(x)=\min\{t>0|\; \mathcal{F}_{\omega}(x,t) \geq0\}, \forall x \in \mathbb{R}^n,
\end{equation}
where
\begin{equation}
\label{eq_10}
\begin{array}{ll}
\mathcal{F}_{\omega}(x,t)&= x^T [(\Lambda(t)-I)^TM(\Lambda(t)-I)- N]x \\
&+ x^T(\Lambda(t)-I)^T M \Omega(t)+\Omega^T(t) M (\Lambda(t)-I)x \\
&+ \Omega^T(t)M\Omega(t).
\end{array}
\end{equation}

Let $\lambda_{\max}^A$ denote $\lambda_{\max}(A+A^T)$ for the sake of compactness. Using Lemma \ref{lem_1} the terms that are dependent on both $x$ and $\Omega(t)$ in $\mathcal{F}_{\omega}(x,t)$ can be decoupled into: 
\begin{equation}
\label{eq_11}
\begin{array}{l}
x^T(\Lambda(t)-I)^T M \Omega(t)+\Omega^T(t) M (\Lambda(t)-I)x \leq \\
\Omega^T(t) \Psi \Omega(t) + x^T(\Lambda(t)-I)^T M \Psi^{-1} M (\Lambda(t)-I)x,
\end{array}
\end{equation}
where $\Psi=\Psi^T \succ0$. Then, it follows that:
\begin{equation}
\label{eq_13}
\begin{array}{l}
\begin{array}{ll}
\Omega^T(t)(M+\Psi)\Omega(t)&
\end{array}\\
\begin{array}{l}
\leq \mu (\int_0^t e^{A(t-s)}E\omega(s)ds)^T (\int_0^t e^{A(t-s)}E\omega(s)ds) \\
 \mbox{(assuming~}M+\Psi\preceq \mu I\mbox{~and~}\mu\geq 0\mbox{)}\\
\leq t \mu \int_0^t e^{(t-s)\lambda_{\max}^A}\omega^T(s)E^TE\omega(s)ds\\
 \mbox{(using Jensen's inequality and Lemma \ref{lem_2})}\\
\leq t W \mu \lambda_{\max}(E^TE) (\int_0^t e^{\lambda_{\max}^A(t-s)}ds) |x|^2\\
 \mbox{(using (\ref{eq_3}))}\\
=tW\mu \lambda_{\max}(E^TE)d_A(t)x^Tx, 
\end{array}
\end{array}
\end{equation}
where
\begin{small}
\begin{equation}
\label{eq_14}
d_A(t)=\left\{
\begin{array}{l}
\frac{1}{\lambda_{\max}^A}(e^{\lambda_{\max}^At}-1),\lambda_{\max}^A\neq 0 \\
t,\lambda_{\max}^A= 0.
\end{array}\right.
\end{equation}
\end{small}

Based on the aforementioned procedure, one concludes that:
\begin{equation}
\label{eq_15}
\mathcal{F}_{\omega}(x,t) \leq x^T \Theta(t)x  
\end{equation}
where
\begin{equation}
\label{eq_16}
\begin{array}{ll}
\Theta(t)&= (\Lambda(t)-I)^T(M+M \Psi^{-1} M)(\Lambda(t)-I)\\
&+tW\mu \lambda_{\max}(E^TE)d_A(t)I- N.
\end{array}
\end{equation}

Then, we employ the Schur complement in order to transform (\ref{eq_16}) into (\ref{eq_17_123}). Note that (\ref{eq_16}) is not linear in $\Psi$ while (\ref{eq_17_123}) is linearly dependent on $\Psi$. Considering (\ref{eq_15}), since $\Phi(t)\succeq 0$ implies $x^T \Theta(t)x\geq 0$ by the Schur complement, it follows that $\tau(x)\geq \tau'(x)$. This concludes the proof.
\end{proof}

Thus, Theorem \ref{theor_1} enables us to avoid the unknown behavior of perturbation $\omega(t)$ in analyzing the sampling begavior of (\ref{eq_8}). However, it is still intractable to use (\ref{eq_17_4}) for the analysis since it has to be checked for an infinite number of instants $t$ and it clearly lacks any insight on how the state $x$ at the sampling instant affects the sampling period $\tau(x)$.

\subsection{Systems and Relations}
\label{sec:modelprel-systems}
In what follows, we review some notions from the field of system theory to formally characterize the outcome of the proposed framework. Let $Z$ be a set and $Q\subseteq Z\times Z$ be an equivalence relation on $Z$. Then, $[z]$ denotes the equivalence class of $z \in Z$ and $Z/Q$ denotes the set of all equivalence classes. A metric (or a distance function) $d:Z \times Z \rightarrow \mathbb{R} \cup \{ + \infty\}$ on $Z$ satisfies, $\forall x, y, z \in Z$: i) $d(x,y)=d(y,x)$, ii) $d(x,y)=0 \leftrightarrow x=y$, and ii) $d(x,y) \leq d(x,z)+d(y,z)$. The ordered pair $(Z,d)$ is said to be a metric space.

\begin{mydefinition}{(Hausdorff Distance \cite{ewald2012combinatorial})}
\label{def_haus}
Assume $X$ and $Y$ are two non-empty subsets of a metric space $(Z,d)$. The Hausdorff distance $d_{H}(X,Y)$ is given by:
$$\max \{ \underset{x \in X}{\sup} \underset{y \in Y}{\inf} d(x,y), \underset{y \in Y}{\sup} \underset{x \in X}{\inf} d(x,y) \}.$$
\end{mydefinition}

It follows that the ordered pair $(\mathbb{IR}^+,d_H)$ is a metric space. Now, we introduce some concepts from system theory and particularly a modified notion of \emph{quotient system} adopted from \cite{Arman2015} (see e.g.\ \cite{Tabuada2009} for the traditional definition). 
\begin{mydefinition}[System {\cite{Tabuada2009}}]
\label{def_sys}
A system is a sextuple $(X,X_0,U,\rTo,Y,H)$ consisting of:
\begin{itemize}
\item a set of states $X$;
\item a set of initial states $X_0 \subseteq X$;
\item a set of inputs $U$;
\item a transition relation $\rTo \subseteq X \times U \times X$;
\item a set of outputs $Y$;
\item an output map $H: X \rightarrow Y$.
\end{itemize}
\end{mydefinition}

When the set of outputs $Y$ of a system is endowed with a metric, it is called a metric system. An autonomous system is a system for which the cardinality of its input set is at most one.

\begin{mydefinition}[Approximate Simulation Relation {\cite{Tabuada2009}}]
\label{def_asr}
Consider two metric systems $S_a=(X_{a},X_{a0},U_{a},\rTo_{a},Y_a,H_{a})$ and $S_b=(X_{b},X_{b0},U_{b},\rTo_{b},Y_b,H_{b})$ with $Y_a=Y_b$, and let $\varepsilon \in \mathbb{R}^+_0$, where $\mathbb{R}^+_0$ represents  the set of nonnegative real numbers. A relation $R \subseteq X_a \times X_b$
is an $\varepsilon$-approximate simulation relation from $S_a$ to $S_b$ if the following three
conditions are satisfied:
\begin{enumerate}
\item $\forall x_{a0} \in X_{a0}, \exists x_{b0} \in X_{b0} \text{ such that } (x_{a0},x_{b0}) \in R$;
\item $\forall (x_a,x_b) \in R, \text{ we have } d(H_a(x_a),H_b(x_b)) \leq \varepsilon$; 
\item $\forall (x_a,x_b) \in R, (x_a,u_a,x'_a) \in \underset{a}{\rTo}$ in $S_a$  $\exists (x_b,u_b,x'_b) \in \underset{b}{\rTo}$ in $S_b$ satisfying $(x'_a,x'_b) \in R$.
\end{enumerate}
We say that $S_b$ $\varepsilon$-approximately simulates $S_a$, denoted by $S_a\preceq^\varepsilon_\mathcal{S} S_b$, if there exists an $\varepsilon$-approximate simulation relation $R$ from $S_a$ to $S_b$. 
\end{mydefinition}

\begin{mydefinition}[Power Quotient System {\cite{Arman2015}}]
\label{def_quo}
Let $S=(X,X_0,\varnothing,\rTo,Y,H)$ be an autonomous system and $R$ be an equivalence relation on $X$.
The power quotient of $S$ by $R$, denoted by $S_{/R}$, is the autonomous system $(X_{/R},X_{/R,0},\varnothing,\rTo_{/R},Y_{/R},H_{/R})$ consisting of:
\begin{itemize}
\item $X_{/R}=X/R$;
\item $X_{/R,0}=\{ x_{/R} \in X_{/R} | x_{/R} \cap X_0 \neq \varnothing \}$; %
\item $(x_{/R},u,x'_{/R}) \in \rTo_{/R}$ if $\exists (x,u,x') \in \rTo$ with $x \in x_{/R}$ and $x' \in x'_{/R}$;
\item $Y_{/R} \subset 2^Y$;
\item $H_{/R}(x_{/R})=\underset{x \in x_{/R}}{\cup}H(x)$.
\end{itemize}
\end{mydefinition}

\begin{mylemma}[{\cite{Arman2015}}]
\label{lem_app_sim}
Let $S$ be an autonomous metric system, $R$ be an equivalence relation on $X$, and let the autonomous metric system $S_{/R}$ be the power quotient system of $S$ by $R$. For any 
$$\varepsilon\geq\max_{\substack{x \in x_{/R} \\ x_{/R}\in X/R}} d(H(x),H_{/R}(x_{/R})),$$ 
with $d$ the Hausdorff distance over the set $2^Y$,
$S_{/R}$ $\varepsilon$-approximately simulates $S$,
i.e.\ $S\preceq^\varepsilon_\mathcal{S} S_{/R}$.
\end{mylemma}

Now, we appropriately modify Definition \ref{def_quo} and Lemma \ref{lem_app_sim} for the case that one can construct an over approximation of the power quotient system, namely $\bar{S}_{/R}$.

\begin{mydefinition}{(Approximate Power Quotient System {\cite{ArmanTCoN}})}
\label{def_quo_app}
Let $\mathcal{S}=(X,X_0,U,\rTo,Y,H)$ be a system, $R$ be an equivalence relation on $X$, and $\mathcal{S}_{/R}=(X_{/R},X_{/R,0},U_{/R},\rTo_{/R},Y_{/R},H_{/R})$ be the power quotient of $\mathcal{S}$ by $R$. An approximate power quotient of $\mathcal{S}$ by $R$, denoted by $\bar{\mathcal{S}}_{/R}$, is a system $(X_{/R},X_{/R,0},U_{/R},\rTo_{\bar{/R}},\bar{Y}_{/R},\bar{H}_{/R})$ such that, $ \underset{\bar{/R}}{\rightarrow} \supseteq \underset{/R}{\rightarrow}$, $\bar{Y}_{/R} \supseteq Y_{/R} $, and $\bar{H}_{/R}(x_{/R}) \supseteq H_{/R}(x_{/R})$, $\forall x_{/R}\in X_{/R}$.
\end{mydefinition}

\begin{mycorollary}[{\cite{ArmanTCoN}}]
\label{cor_1}
Let $S$ be a metric system, $R$ be an equivalence relation on $X$, and let the metric system $\bar{S}_{/R}$ be the approximate power quotient system of $S$ by $R$. For any 
$$\varepsilon\geq\max_{\substack{x \in x_{/R} \\ x_{/R}\in X_{/R}}} d(H(x),\bar{H}_{/R}(x_{/R})),$$ 
with $d$ the Hausdorff distance over the set $2^Y$,
$\bar{S}_{/R}$ $\varepsilon$-approximately simulates $S$,
i.e.\ $S\preceq^\varepsilon_\mathcal{S} \bar{S}_{/R}$.
\end{mycorollary}

\subsection{Timed Safety Automaton}

In what follows, we present a formal definition for TSA. A TSA \cite{Alur1994} is a directed graph extended with real-valued variables (called clocks) that model the logical clocks. We define $C$ as a set of finitely many clocks. Clock constraints are used to restrict the behavior of the automaton. A clock constraint is a conjunctive formula of atomic constraints of the form $x \bowtie n$ or $x - y \bowtie n$ for $x,y \in  C$, $\bowtie \in \{\leq,<,=,>,\geq\}$ and $n \in \mathbb N$. We use $\mathcal B(C)$ to denote the set of clock constraints.

\begin{mydefinition}\label{def:ta}(Timed Safety Automaton \cite{henzinger1994symbolic})
A \emph{timed safety automaton} $\mathsf{TSA}$ is a sextuple $(L,\ell_0,\mathsf{Act},C,E,\mathsf{Inv})$ where:
\begin{itemize}
\item $L$ is a set of finitely many locations (or vertices);
\item $\ell_0 \in L$ is the initial location;
\item $\mathsf{Act}$ is the set of actions;
\item $C$ is a set of finitely many real-valued clocks;
\item $E \subseteq L \times \mathcal B(C) \times \mathsf{Act} \times 2^{C} \times L$ is the set of edges;
\item $\mathsf{Inv} : L \to \mathcal B(C)$ assigns invariants to locations.
\end{itemize}

The location invariants are restricted to constraints of the form: $c \leq n$ or $c < n$ where $c$ is a clock and $n$ is a natural number.
\end{mydefinition}

\subsection{Problem Statement}
\label{sec:abs-problem}
Consider the system 
$S=(X,X_0,\varnothing,\rTo,Y,H)$:
\begin{itemize}
\item $X=\mathbb{R}^n$;
\item $X_0=\mathbb{R}^n$;
\item $(x,x') \in \rTo$ iff $\xi_x(\tau(x))=x'$ given by \eqref{eq_1}-\eqref{eq_2}, and \eqref{eq_8};
\item $Y \subset \mathbb{R}^+$;
\item $H: \mathbb{R}^n \rightarrow \mathbb{R}^{+}$ where $H(x)=\tau(x)$.
\end{itemize}
The output of the above system generates all possible sequences of inter-sample intervals of the concrete system \eqref{eq_1}-\eqref{eq_2} with the TM~\eqref{eq_8}.

\begin{myproblem}
\label{prob_1}
Provide a construction of power quotient systems $S_{/\mathcal{P}}$ of systems $S$ as defined above.
\end{myproblem}

Based on Definition \ref{def_quo}, we propose to construct the system 
$S_{/\mathcal{P}}=(X_{/\mathcal{P}},X_{/\mathcal{P},0},\varnothing, \underset{/\mathcal{P}}{\rTo},Y_{/\mathcal{P}},H_{/\mathcal{P}})$
where 
\begin{itemize}
\item $X_{/\mathcal{P}}= \mathbb{R}_{/\mathcal{P}}^n:= \{\mathcal{R}_1,\dots,\mathcal{R}_q \}$;
\item $X_{/\mathcal{P},0}=\mathbb{R}_{/\mathcal{P}}^n$;
\item $(x_{/\mathcal{P}},x'_{/\mathcal{P}}) \in
\underset{/\mathcal{P}}{\rTo} $ if $\exists x \in x_{/\mathcal{P}}$, $\exists x' \in x'_{/\mathcal{P}}$ such that $\xi_x(H(x))=x'$ as determined by \eqref{eq_1}-\eqref{eq_2};
\item $Y_{/\mathcal{P}} \subset 2^Y \subset \mathbb{IR}^{+}$, where $\mathbb{IR}^{+}$ represents the set of closed intervals $[a,b]$ such that $0 < a \leq b$;
\item \small{$H_{/\mathcal{P}}(x_{/\mathcal{P}})= [\underset{x \in x_{/\mathcal{P}}}{\min}H(x),\underset{x \in x_{/\mathcal{P}}}{\max}H(x)]:=[\ubar{\tau}_{x_{/\mathcal{P}}},\bar{\tau}_{x_{/\mathcal{P}}}]$}.
\end{itemize}

The equivalence relation $\mathcal{P}$ on $\mathbb{R}^n$ partitions the state space of $S$ (i.e. the ETC system) into the set $X_{/\mathcal{P}}$ with a finite cardinality. However, since the exact construction of $\mathcal{S}_{/\mathcal{P}}$ is in general impossible, we construct instead $\bar{\mathcal{S}}_{/\mathcal{P}}$ (see Definition~\ref{def_quo_app}). Later on, it will be shown that the constructed $\bar{S}_{/\mathcal{P}}$ is equivalent to a TSA.

\section{Abstractions of event-triggered LTI systems}
\label{sec:abs}

In this section, we introduce the framework to solve Problem~\ref{prob_1} in the following order: 1) a proper definition of an equivalence relation $\mathcal{P}$ on $\mathbb{R}^n$, 2) a tractable approach to compute the output map $\bar{H}_{/\mathcal{P}}$ and its corresponding output set $\bar{Y}_{/\mathcal{P}}$, and 3) a reachability-based analysis to derive the discrete transitions among abstract states $x_{/\mathcal{P}}$. 

\subsection{State set}
\label{sec:abs-construction-states}

The type of state set construction approach mainly relies on an intuitive observation from \eqref{eq_15}.

\begin{myremark}
\label{rem_1} 
Consider that the right-hand side of (\ref{eq_15}) is used to analyze the sampling behavior of (\ref{eq_17}) instead of $\mathcal{F}_{\omega}(x,t)$. Then, the sampling periods of all states, located on a line that passes through the origin excluding the origin itself, are lower bounded by the same sampling period, i.e.\ $\tau'(x)=\tau'(\lambda x)$, $\forall \lambda \neq 0$.
\end{myremark}

It follows that a proper approach to abstract the state space is via partitioning it into a finite number of convex polyhedral cones (pointed at the origin) $\mathcal{R}_s$ where $s \in \{1,\dots,q \}$ and $\bigcup_{s=1}^q \mathcal{R}_s = \mathbb{R}^n$. This state space abstraction technique is proposed by \cite{Fiter2012}, dividing each of the angular spherical coordinates of $x \in \mathbb{R}^n$: $\theta_1, \dots,\theta_{n-2} \in [0,\pi]$, $\theta_{n-1} \in [-\pi,\pi]$ into $\bar{m}$ (not necessarily equidistant) intervals resulting in $q=\bar{m}^{(n-1)}$ conic regions. Furthermore, since the term $x^T \Theta(t)x$ is quadratic in $x$, it is sufficient to only analyze half of the state space (e.g.\ by taking $\theta_{n-1} \in [0,\pi]$). Thus, the equivalence relation $\mathcal{P}$ to construct the abstraction is given by:
$$
(x,x')\in\mathcal{P} \Leftrightarrow \exists\, s \in  \{1,\dots,q \}\, \text{s.t.}\,  x,x'\in \mathcal{R}_s,
$$
where $q$ is the number of equivalence classes. Hence, the equivalence classes of $\mathcal{P}$ are defined by polyhedral cones pointed at the origin given by $\mathcal{R}_s=\{ x \in \mathbb{R}^2| \; x^T Q_s x \geq 0 \},\; Q_s\in \mathcal{S}_2$ whenever $n=2$ or $\mathcal{R}_s=\{ x \in \mathbb{R}^n| \; E_s x \geq 0\}$, $E_s \in \mathcal{S}_{n\times p}$ otherwise.

\subsection{Output Map}
\label{sec:abs-construction-output}

In this subsection, we present how to construct $\bar{H}_{/\mathcal{P}}$ and $\bar{Y}_{/\mathcal{P}}$. For all $x \in \mathcal{R}_s$, the output $\bar{Y}_{/\mathcal{P}}=\bar{H}_{/\mathcal{P}}(x)$ is equal to the time interval $[\ubar{\tau}_s,\bar{\tau}_s]$ indicating $\tau(x) \in [\ubar{\tau}_s,\bar{\tau}_s]$. We make use of the polytopic embedding technique proposed by \cite{hetel2006stabilization}. In the space of real matrices, a sequence of convex polytopes is constructed around the matrix $\Phi(t)$. Doing so replaces the evaluation of (\ref{eq_17_4}) at infinitely many instants $t$ by the evaluation of $\Phi_{\kappa,s}$ at finitely many vertices in the sequence of polytopes generated by $\Phi_{\kappa,s}$. Assume a scalar $\sigma>0$ denoting a time instant for which the TM \eqref{eq_8} is enabled in the whole state space, i.e.\ $\Phi(t)\succeq 0$. Consider $N_{\text{conv}}+1$ is the number of vertices employed to define the polytope containing $\Phi(t)$ in a given time interval,  and $l\geq1$ denotes the number of time subdivisions considered in the time interval $[0,\sigma]$.

\begin{mylemma}
\label{lem3}
Let  $s \in \{1,\dots,q\}$. Consider a time instant $\ubar{\tau}_s \in (0,\sigma]$, a scalar $\mu$ and a symmetric matrix $\Psi$ satisfying (\ref{eq_17_1_a}). If $ \ubar{\Phi}_{(i,j),s}  \preceq 0$ holds $\forall(i,j) \in \mathcal{K}_s=(\{0,\dots,N_{\text{conv}} \} \times \{0,\dots,\lfloor \frac{\ubar{\tau}_s l}{\sigma} \rfloor \})$, then, it follows that $ \Phi(t)\preceq 0$, $\forall t \in [0,\ubar{\tau}_s]$ with $\Phi$ defined in \eqref{eq_17_123} and
\begin{equation*}
\label{eq_34}
\ubar{\Phi}_{(i,j),s}=\tilde{\ubar{\Phi}}_{(i,j),s}+\eta I
\end{equation*}
\begin{equation}
\label{eq_24}
\tilde{\ubar{\Phi}}_{(i,j),s}=\left\{
\begin{array}{lll}
\Sigma_{k=0}^i \hat{\ubar{\Phi}}_{(i,j),s} (\frac{\sigma}{l})^k &,&j<\lfloor \frac{\ubar{\tau}_s l}{\sigma} \rfloor \\
\Sigma_{k=0}^i \hat{\ubar{\Phi}}_{(i,j),s} (\ubar{\tau}_s-j\frac{\sigma}{l})^k&,&j=\lfloor \frac{\ubar{\tau}_s l}{\sigma} \rfloor,
\end{array}\right.
\end{equation} 
\begin{equation}
\label{eq_24_1}
\begin{array}{l}
\hat{\ubar{\Phi}}_{(0,j),s}=\left[ \begin{array}{lr}
\ubar{L}_{0,j} & \check{\Pi}_j^T M^T\\ 
M \check{\Pi}_j & -\Psi
\end{array}\right],\\
\hat{\ubar{\Phi}}_{(k\geq 1,j),s}=\left[ \begin{array}{lr}
\ubar{L}_{k,j} & \hat{\Pi}_j^T \frac{(A^{k-1})^T}{k!} M^T\\ 
M \frac{A^{k-1}}{k!} \hat{\Pi}_j & 0
\end{array}\right],
\end{array}
\end{equation}
and
\begin{equation}
\label{eq_25}
\ubar{L}_{0,j}=\check{\Pi}_j^T M \check{\Pi}_j- N+\tilde{L}_{0,j}
\end{equation}
with
\begin{equation}
\label{eq_26}
\tilde{L}_{0,j}=\left\{
\begin{array}{l}
W\mu \frac{\lambda_{\max}(E^TE)}{\lambda_{\max}^A}(j\frac{\sigma}{l})(e^{\lambda_{\max}^Aj\frac{\sigma}{l}}-1)I \\
\quad \mbox{for~}\quad\lambda_{\max}^A\neq 0, \\
W\mu \frac{\lambda_{\max}(E^TE)}{\lambda_{\max}^A}(j\frac{\sigma}{l})^2I\\
\quad \mbox{for~}\quad\lambda_{\max}^A= 0,
\end{array}\right.
\end{equation}
\begin{equation}
\label{eq_27}
\ubar{L}_{1,j}=\check{\Pi}_j^T M \hat{\Pi}_j+ \hat{\Pi}_j^T M \check{\Pi}_j+\tilde{L}_{1,j}
\end{equation}
with
\begin{small}
\begin{equation}
\label{eq_28}
\tilde{L}_{1,j}=\left\{
\begin{array}{l}
W\mu \frac{\lambda_{\max}(E^TE)}{\lambda_{\max}^A}[(j\frac{\sigma}{l})e^{\lambda_{\max}^Aj\frac{\sigma}{l}} \lambda_{\max}^A\\
+e^{\lambda_{\max}^Aj\frac{\sigma}{l}}-1]I \mbox{~ for ~} \lambda_{\max}^A\neq 0, \\
W\mu (2j\frac{\sigma}{l}) \lambda_{\max}(E^TE) I \mbox{~ for ~} \lambda_{\max}^A= 0,
\end{array}\right.
\end{equation}
\end{small}
\begin{equation}
\label{eq_29}
\ubar{L}_{2,j}=\check{\Pi}_j^T M \frac{A}{2!} \hat{\Pi}_j+ \hat{\Pi}_j^T \frac{A^T}{2!} M \check{\Pi}_j+\hat{\Pi}_j^T M\hat{\Pi}_j+\tilde{L}_{2,j}
\end{equation}
with
\begin{small}
\begin{equation}
\label{eq_30}
\tilde{L}_{2,j}=\left\{
\begin{array}{l}
W\mu \frac{\lambda_{\max}(E^TE)}{\lambda_{\max}^A}[(j\frac{\sigma}{l})e^{\lambda_{\max}^Aj\frac{\sigma}{l}} \frac{(\lambda_{\max}^A)^2}{2!}\\+e^{\lambda_{\max}^Aj\frac{\sigma}{l}} \lambda_{\max}^A]I\\
\quad \mbox{~for~}\quad \lambda_{\max}^A\neq 0, \\
W\mu \lambda_{\max}(E^TE)I \quad \mbox{~for~}\quad \lambda_{\max}^A= 0,
\end{array}\right.
\end{equation}
\end{small}
\begin{equation}
\label{eq_31}
\begin{array}{ll}
\ubar{L}_{k\geq3,j}&=\check{\Pi}_j^T M \frac{A^{k-1}}{k!} \hat{\Pi}_j+\hat{\Pi}_j^T \frac{(A^{k-1})^T}{k!} M\check{\Pi}_j\\
&+ \hat{\Pi}_j^T (\Sigma_{i=1}^{k-1} \frac{(A^{i-1})^T}{i!} M\frac{A^{k-i-1}}{(k-i)!}) \hat{\Pi}_j+\tilde{L}_{k,j}
\end{array}
\end{equation}
with
\begin{small}
\begin{equation}
\label{eq_32}
\tilde{L}_{k\geq3,j}=\left\{
\begin{array}{l}
W\mu \frac{\lambda_{\max}(E^TE)}{\lambda_{\max}^A}[(j\frac{\sigma}{l})e^{\lambda_{\max}^Aj\frac{\sigma}{l}} \frac{(\lambda_{\max}^A)^k}{k!}\\
+e^{\lambda_{\max}^Aj\frac{\sigma}{l}} \frac{(\lambda_{\max}^A)^{k-1}}{(k-1)!}]I\\ \quad \mbox{~for~}\quad  \lambda_{\max}^A\neq 0, \\
0 \quad \mbox{~for~}\quad  \lambda_{\max}^A= 0,
\end{array}\right.
\end{equation}
\end{small}
\begin{small}
\begin{equation}
\label{eq_33}
\eta \geq \underset{t' \in [0,\frac{\sigma}{l}],r \in \{0,...,l-1 \}}{\max} \lambda_{\max} \left( \Phi(t'+r\frac{\sigma}{l})-\Sigma_{k=0}^N \hat{\ubar{\Phi}}_{k,r}(t')^k \right),
\end{equation}
\end{small}
and
\begin{equation}
\label{my_etc25}
\begin{array}{l}
\tilde{\ubar{\Phi}}_{(N_{\text{conv}},j)}(t')=\Sigma_{k=0}^{N_{\text{conv}}} \hat{\ubar{\Phi}}_{k,j}(t')^k.
\end{array}
\end{equation}
\end{mylemma}

\begin{proof}
See Appendix. 
\end{proof}

Then, using the S-procedure, the following theorem provides an approach to regionally reduce the conservatism involved in the $\ubar{\tau}_s$ estimates obtained from Lemma~\ref{lem3}. 

\begin{mytheorem}[Regional Lower Bound Approximation]
\label{them3}
Consider a scalar $\ubar{\tau}_s\in(0, \sigma]$, a scalar $\mu$ and a symmetric matrix $\Psi$ satisfying (\ref{eq_17_1_a}), and matrices $\ubar{\Phi}_{\kappa,s}$, $\kappa=(i,j)\in\mathcal{K}_s$, defined as in Lemma~\ref{lem3}. If there exist scalars $\ubar{\alpha}_{\kappa,s}\geq 0$ (for $n=2$) or symmetric matrices $\ubar{U}_{\kappa,s}$ with nonnegative entries (for $n \geq 3$) such that for all $\kappa\in\mathcal{K}_s$ the following LMIs hold:
\begin{equation}
\label{eq_them_un}
\begin{array}{l}
\left\{\begin{array}{ll}
\ubar{\Phi}_{(i,j),s}+\left[\begin{array}{lr}
\ubar{\alpha}_{(i,j),s} Q_s &0\\
0&0
\end{array}\right]
 \preceq 0 & \mbox{if } n=2,\\
\ubar{\Phi}_{(i,j),s}+\left[\begin{array}{lr}
E_s^{T} \ubar{U}_{(i,j),s} E_s &0\\
0&0
\end{array}\right] \preceq 0  & \mbox{if } n\geq 3,
\end{array}\right.
\end{array}
\end{equation}
then, the inter-sample time \eqref{eq_8} of the system \eqref{eq_1}-\eqref{eq_2} is regionally bounded from below by $\ubar{\tau}_s,\; \forall x \in \mathcal{R}_s$.
\end{mytheorem}

\begin{proof}
See Appendix.
\end{proof}

One can follow a similar approach to find the upper bounds $\bar{\tau}_s$ on the inter-sample times that is outlined in Lemma \ref{lemm_u1} and Theorem \ref{them4}.

\begin{mylemma}
\label{lemm_u1}
Let $s \in \{ 1, \ldots, q \}$. Consider a time instant $\bar{\tau}_s \in [\ubar{\tau}_s,\sigma]$, a scalar $\mu$ and a matrix $\Psi$ satisfying the LMI conditions given in Lemma \ref{lem3}. If
$\bar{\Phi}_{(i,j),s} \preceq 0$ holds $\forall (i,j) \in \mathcal{K}_s=(\{0,\dots,N_{\text{conv}} \} \times \{ \lfloor \frac{\bar{\tau}_s l}{\sigma} \rfloor ,\dots, l-1 \})$,
then, it follows that $ \Phi(t) \succeq 0$, $\forall t \in [\bar{\tau}_s,\sigma]$ with $\Phi$ defined in \eqref{eq_17_123} and
\begin{equation*}
\begin{array}{l}
\bar{\Phi}_{(i,j),s}=-\bar{\tilde{\Phi}}_{(i,j),s}- \eta I, 
\end{array}
\end{equation*}
\begin{equation*}
\begin{array}{l}
\bar{\tilde{\Phi}}_{(i,j),s}=\left\{\begin{array}{ll}
\sum_{k=0}^i \; L_{k,j} (\frac{(j+1) \sigma}{l} - \bar{\tau}_s)^k & \text{if}  \; j = \lfloor \frac{\bar{\tau}_s l}{\sigma} \rfloor, \\
\sum_{k=0}^i \; L_{k,j} (\frac{\sigma}{l})^k  & \text{if} \; j > \lfloor \frac{\bar{\tau}_s l}{\sigma} \rfloor,
\end{array}\right. 
\end{array}
\end{equation*}
where $L_{k,j}$ are given by \eqref{eq_25}-\eqref{eq_32} and $\eta$ is defined in (\ref{eq_33}). 
\end{mylemma}

\begin{proof}
See Appendix.
\end{proof}

\begin{mytheorem}[Regional Upper Bound Approximation]
\label{them4}
Consider a scalar $\bar{\tau}_s\in [\ubar{\tau}_s,\sigma]$, a scalar $\mu$ and a symmetric matrix $\Psi$ satisfying (\ref{eq_17_1_a}), and matrices $\bar{\Phi}_{\kappa,s}$, $\kappa=(i,j) \in \mathcal{K}_s$, defined as in Lemma~\ref{lemm_u1}. If there exist scalars $\bar{\alpha}_{\kappa,s}\geq 0$ (for $n=2$) or symmetric matrices $\bar{U}_{\kappa,s}$ with nonnegative entries (for $n \geq 3$) such that for all $\kappa\in\mathcal{K}_s$ the following LMIs hold:
\begin{equation}
\begin{array}{l}
\left\{\begin{array}{ll}
\bar{\Phi}_{(i,j),s}-\left[\begin{array}{lr}
\bar{\alpha}_{(i,j),s} Q_s &0\\
0&0
\end{array}\right]
 \preceq 0 & \mbox{if } n=2,\\
\bar{\Phi}_{(i,j),s}-\left[\begin{array}{lr}
E_s^{T} \bar{U}_{(i,j),s} E_s &0\\
0&0
\end{array}\right] \preceq 0  & \mbox{if } n\geq 3,
\end{array}\right.
\end{array}
\end{equation}
then, the inter-sample time \eqref{eq_8} of the system \eqref{eq_1}-\eqref{eq_2} is regionally bounded from above by $\bar{\tau}_s,\; \forall x \in \mathcal{R}_s$.
\end{mytheorem}

\begin{proof}
Analogous to the proof of Theorem~\ref{them3}. \end{proof}

\subsection{Transition Relations}
\label{sec:abs-construction-trans}
In order to find all the transitions in $\bar{S}_{/\mathcal{P}}$, it is required to compute the reachable set of each $\mathcal{R}_s$ over the time interval $[\ubar{\tau}_s,\bar{\tau}_s]$. In the sequel, we present how one is able to compute over approximations of the reachable set of each cone by the Minkowski sum of two sets. The evolution of states over this time interval is given by $\xi_x(\tau)=\Lambda(\tau)x+\Omega(\tau)$. Denote by $\mathcal{X}_{[\ubar{\tau}_s,\bar{\tau}_s]} (X_{0,s})$ the reachable set of $X_{0,s}$ during the time interval $[\ubar{\tau}_s,\bar{\tau}_s]$, that is given by: 
\begin{equation*}
\{x'\in\mathbb{R}^n\,|\,\exists x\in X_{0,s}, \exists \tau\in [\ubar{\tau}_s,\bar{\tau}_s], x'=\xi_x(\tau)\}.
\end{equation*}

Furthermore, define 
\begin{equation*}
\begin{array}{l}
\mathcal{X}^1_{[\ubar{\tau}_s,\bar{\tau}_s]}(X_{0,s})\\
:=\{x'\in\mathbb{R}^n\,|\,\exists x\in X_{0,s}, \exists \tau\in [\ubar{\tau}_s,\bar{\tau}_s], x'=\Lambda(\tau)x\},\\
\mathcal{X}^2_{[\ubar{\tau}_s,\bar{\tau}_s]}(X_{0,s})\\
:=\{x'\in\mathbb{R}^n\,|\,\exists x\in X_{0,s}, \exists \tau\in [\ubar{\tau}_s,\bar{\tau}_s], x'=\Omega(\tau)\}.
\end{array}
\end{equation*}
It follows that:
 
$\mathcal{X}_{[\ubar{\tau}_s,\bar{\tau}_s]}(X_{0,s}):=\mathcal{X}^1_{[\ubar{\tau}_s,\bar{\tau}_s]}(X_{0,s})\bigoplus\mathcal{X}^2_{[\ubar{\tau}_s,\bar{\tau}_s]}(X_{0,s})$. 
\vspace{1mm}

In \cite[Section III.B.3]{Arman2015}, it has been shown that it is enough to consider subsets $X_{0,s}\subset \mathcal{R}_s$ being convex polytopes with each vertex placed on each of the extreme rays of $\mathcal{R}_s$ (excluding the origin) to compute $\mathcal{X}^1_{[\ubar{\tau}_s,\bar{\tau}_s]}(X_{0,s})$. Then, one can effectively compute an over approximation of the reachable set of a polytope under linear time invariants, denoted by $\hat{\mathcal{X}}_{[\ubar{\tau}_s,\bar{\tau}_s]}(X_{0,s})$, see e.g.\ \cite{Chutinan1998}. Furthermore, one has: 
\begin{equation*}
\label{eq_reach}
\begin{array}{ll}
\|\Omega(\tau) \|&=\|\int_0^{\tau} e^{A(\tau-s)}E\omega(s)ds \|\\
&\leq \int_0^{\tau} \|e^{A(\tau-s)}E\omega(s) \|ds\\
&\leq \int_0^{\tau}\|e^{A(\tau-s)} \| \|E \| |\omega(s)|ds\\
&W|x | \|E \|\int_0^{\tau} |e^{\mu(A)(\tau-s)}|ds\\
&=\rho(\tau) |x |
\end{array}
\end{equation*}
where $\rho(\tau)= W \|E \|\int_0^{\tau} |e^{\mu(A)(\tau-s)}|ds$. Thus, it follows that $\mathcal{X}^2_{[\ubar{\tau}_s,\bar{\tau}_s]}(X_{0,s})$ can be over approximated by
a second order cone given by:
\begin{equation*}
\begin{array}{l}
\hat{\mathcal{X}}^2_{[\ubar{\tau}_s,\bar{\tau}_s]}(X_{0,s})\\
:=\{ x'\in \mathbb{R}^n| \; \exists x\in X_{0,s}, \exists \tau \in [\ubar{\tau}_s,\bar{\tau}_s],| x'|\leq \rho(\bar{\tau}_s)|x| \}.
\end{array}
\end{equation*} 

To compute the transitions in $\bar{S}_{/\mathcal{P}}$, it thus suffices to derive the intersection between the over approximation $\hat{\mathcal{X}}_{[\ubar{\tau}_s,\bar{\tau}_s]}(X_{0,s})$ and all the conic regions $\mathcal{R}_{t}$ where $t\in\{1,\dots,q\}$. To compute transitions, it is required to check whether the following convex feasibility problem for each conic region $\mathcal{R}_{t}$ holds:
\begin{equation}
\mathcal{R}_t\cap\hat{\mathcal{X}}_{[\ubar{\tau}_s,\bar{\tau}_s]}(X_{0,s})\neq\varnothing,
\label{my_etc51}
\end{equation}
which can be solved by existing convex analysis tools. There exists a transition from abstract state $\mathcal{R}_s$ to $\mathcal{R}_t$ in $\bar{S}_{/\mathcal{P}}$ in the case that \eqref{my_etc51} is satisfied. 

\subsection{Timed Safety Automata Representation}
\label{sec:abs-construction-ta}
In this subsection, first, we point out the connection between an abstract state $x_{/\mathcal{P}} \in X_{/\mathcal{P}}$ and its corresponding output $y_{/\mathcal{P}} \in \bar{Y}_{/\mathcal{P}}$ \cite{Arman2015}. The system $\bar{S}_{/\mathcal{P}}$: 
\begin{enumerate}
\item remains at $x_{/\mathcal{P}}$ during the time interval $[0,\ubar{\tau}_{x_{/\mathcal{P}}})$, 
\item possibly leaves $x_{/\mathcal{P}}$ during the time interval $[\ubar{\tau}_{x_{/\mathcal{P}}},\bar{\tau}_{x_{/\mathcal{P}}})$, and
\item is forced to leave $x_{/\mathcal{P}}$ at the time instant $\bar{\tau}_{x_{/\mathcal{P}}}$.
\end{enumerate}

Thus, the semantics of $\bar{S}_{/\mathcal{P}}$ is equivalent to a timed safety automaton given by $\mathsf{TSA} = (L,\ell_0,\mathsf{Act},C,E,\mathsf{Inv})$ where: 
\begin{itemize}
\item $L=X_{\bar{/\mathcal{P}}}$;
\item $\ell_0 := \mathcal{R}_s$ such that $\xi(0) \in \mathcal{R}_s$;
\item $\mathsf{Act} = \{*\}$ is an arbitrary symbol;
\item $C=\{c\}$; 
\item $E$ is given by all tuples $(\mathcal{R}_s,g,a,r,\mathcal{R}_t)$ such that $(\mathcal{R}_s,\mathcal{R}_t)\in \rTo_{\bar{/\mathcal{P}}}$, $g=\{c | \; c \in [\ubar{\tau}_s, \bar{\tau}_s] \}$, $a=*$, and $r$ is given by $c:=0$;
\item $\mathsf{Inv}(\mathcal{R}_s):=\{c| c \in [0,\bar{\tau}_s]\}, \forall s \in \{1,\dots,q\}$.
\end{itemize}

Finally, in this section, we state the following fact. Although the construction technique presented in this section is offline, it is  exponentially dependent on $n-1$ (where $n$ is the number of states) and hence it is computationally expensive for higher-order systems.

\section{Numerical Example}
We illustrate the theoretical results of this paper in a numerical example. Consider an LTI sytem, used as an example in \cite{Tabuada}, and add a perturbation term $\omega(t)$ as follows:   
\begin{equation}
\begin{array}{l}
\dot\xi(t) =
\begin{bmatrix}
0 & 1\\
-2 & 3
\end{bmatrix}
\xi(t) +
\begin{bmatrix}
0 \\ 1
\end{bmatrix}
\nu(t)+
\begin{bmatrix}
0 \\ 1
\end{bmatrix}
\omega(t)
\end{array}
\label{eqn:lti-ex-1}
\end{equation}
with the perturbation bound $W=0.001$. We set the scalars, associated with $\mathcal{L}_2$-based TM's, $\gamma=100$, $\beta=0.25$, see (\ref{are_1})-(\ref{are_2}) and (\ref{eq_8})-(\ref{eq_9}). Then, solving the ARE associated with the $\mathcal{L}_2$ stability, the control update law (implemented in a sample-and-hold fashion) is computed, that is, $ \forall t\in[t_k,t_{k+1})$: 
\begin{equation*}
\label{eqn:lti-ex-2}
\nu(t)=-K\xi(t_k)=-[0.2361\quad 6.2367]\xi(t_k),
\end{equation*}
where $t_k$ denotes the sampling instants and $k\in \mathbb{N}_0$. Now, we set the order of polynomial approximation $N_{\text{conv}}=7$, the number of polytopic subdivisions $l=800$, the upper bound of the inter-sample intervals $\sigma=8$, the number of angular sub-divisions $\bar{m}=10$, thus, $q=2\times 10^{(2-1)}=20$. Then, applying the results from Section \ref{sec:abs-construction-output}, we get the precision abstraction of $\varepsilon=6.100$. Compared to the results found in \cite{ArmanTCoN}, the derived $\varepsilon$ is large. However, one must take into account the possible stabilizing effect of disturbance on the dynamics in (\ref{eq_1}) can enlarge the derived $\bar{\tau}_s$ and a more thorough study is due in this regard. In Figure \ref{fig_1}, the derived lower and upper bounds are depicted. It is evident that the derived $\ubar{\tau}_s$ compared to the MIET are less conservative and can be effectively used for scheduling.

\begin{figure}
\centering
\includegraphics[width=1\columnwidth]{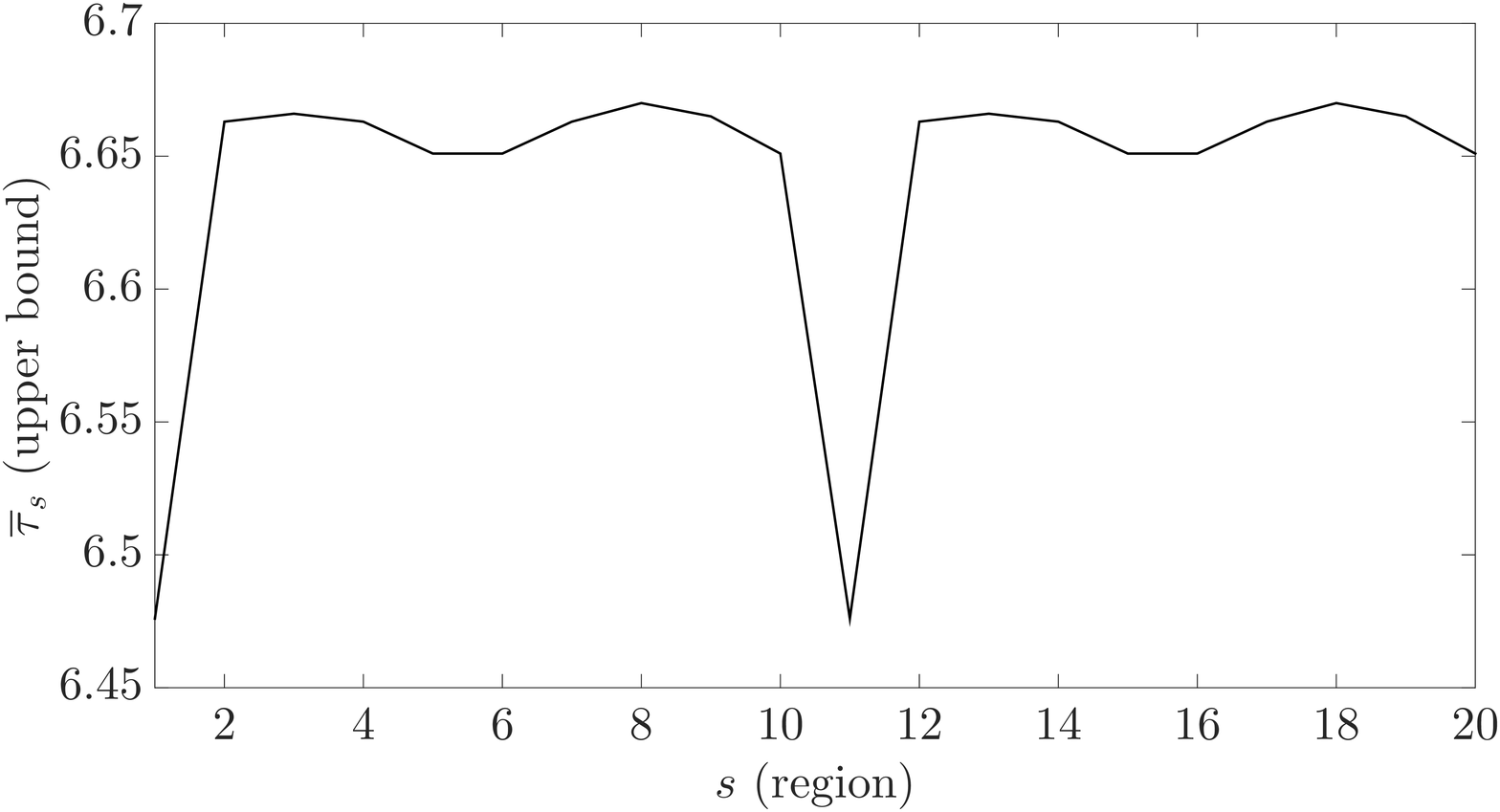}
\includegraphics[clip=true,width=1\columnwidth]{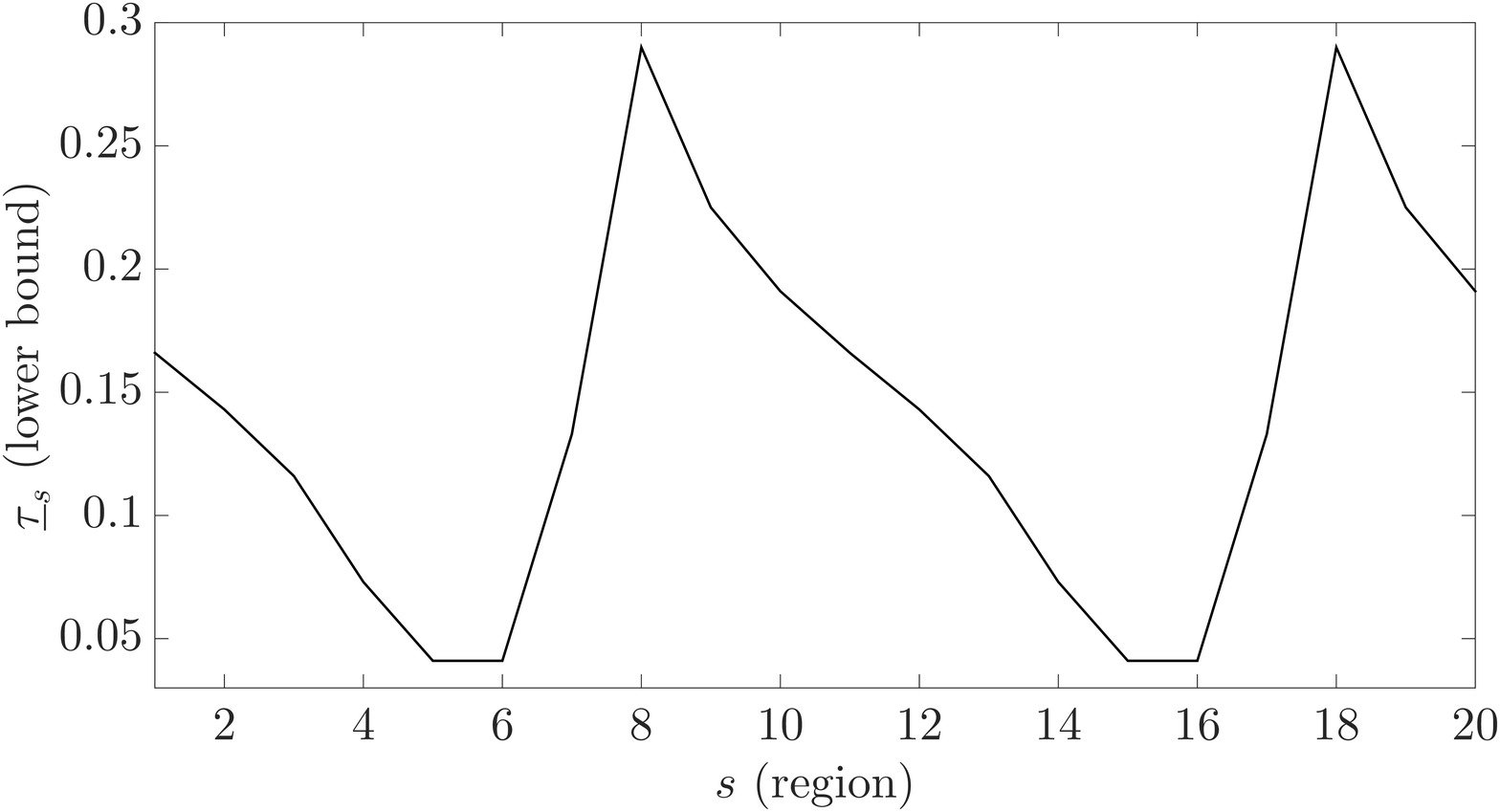}
\caption{(Top plot) Upper bounds on regional inter-sample times. (Bottom plot) Lower bounds on regional inter-sample times.}
\label{fig_1}
\end{figure}

Figure \ref{fig_2} represents the conic regions $s$ and the associated $\ubar{\tau}_s$ and $\bar{\tau}_s$ (note that in order to show the lower bounds in a clear manner the lower and upper bounds are depicted, separately).

\begin{figure}[thpb]
\centering
\includegraphics[width=1\columnwidth]{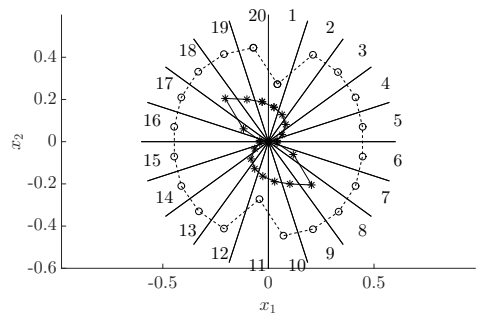}
\caption{The radial distance from the origin of each asterisk indicates the regional lower bound of the indexed cone. Furthermore, in the case of circles, the distance indicates the regional upper bound of the indexed cone minus $6.2\sec$, i.e.\ $\bar{\tau}_s-6.2\sec$ (for the sake of clarity of the figure).}
\label{fig_2}
\end{figure}

Moreover, Figure~\ref{fig_3} represents the simulation of the control system for a simulation time of $15\sec$. It is clear that the bounds derived by the analysis given in Section \ref{sec:abs-construction-output} have been respected by the sampling periods generated by the control system.

\begin{figure}[thpb]
\centering
\includegraphics[width=1\columnwidth]{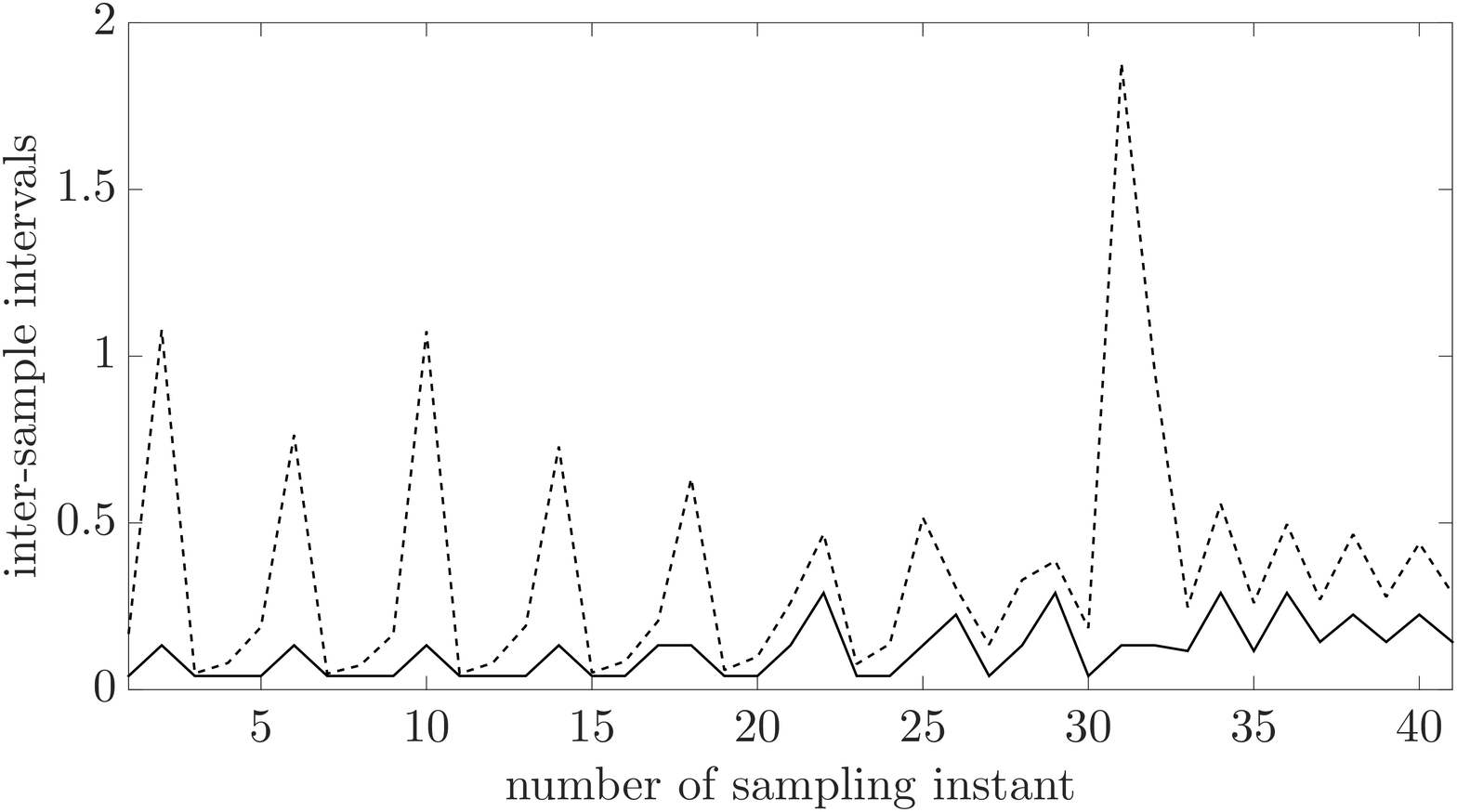}
\caption{Schematic representation of validation of lower bounds during the simulation period, the solid line (dashed line) represents the lower bounds on inter-sample intervals  (generated inter-sample intervals during simulation).}
\label{fig_3}
\end{figure}

Figure~\ref{fig_4} depicts the result of applying the procedure introduced in Section \ref{sec:abs-construction-trans}. 

\begin{figure}[thpb]
\centering
\includegraphics[width=1\columnwidth]{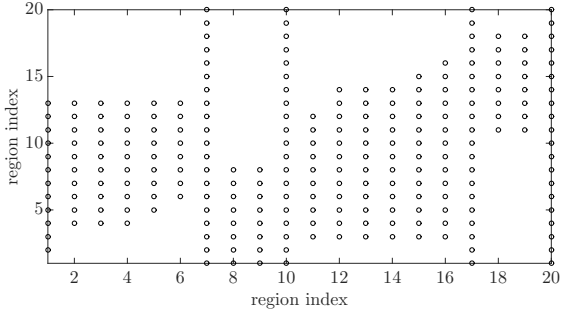}
\caption{Schematic representation of set of edges in the timed automaton generated by \eqref{eqn:lti-ex-1}-\eqref{eqn:lti-ex-2}. A circle at the coordinate $(i,j)$ denotes an edge from location $i$ to location $j$.}
\label{fig_4}
\end{figure}

\section{Conclusions and Future Work}
We have presented an approach to capture the sampling behavior of perturbed LTI systems with $\mathcal{L}_2$-based triggering mechanisms by timed safety automata. It has been shown that the derived timed automaton $\varepsilon$-approximately simulates the ETC system. The main contribution of this study falls into the subject of synthezing scheduling policies for ETC feedback loops using timed automata. Because of the inherent robustness of ETC strategies to perturbations, they are more appealing in practical applications compared to STC strategies. However, most of the existing ETC strategies are equipped with the quantity of minimum inter-sample time that indicates the maximum utilization of communication bandwidth. Despite the fact that such quantity can practically be used in scheduling of ETC feedback loops, it does not enable the full exploitation of the beneficiary features of ETC strategies in scheduling feedback policies. In fact, using solely such quantity in scheduling policies results in TDC-like techniques. Furthermore, in most of decentralized ETC strategies in the literature, the existence of a minimum inter-sample time among different subsystems is absent, see e.g.\ \cite{donkers2012output}. Exploiting the already existing tools for the synthesis of timed automata, one can further extend the results of this study to synthesize conflict-free policies in WNCS's, see e.g.\ \cite{Arman2015} which proposed a centralized scheduling of feedback policies. Another promising direction to follow is to find a fully decentralized approach instead of the centralized approach proposed in \cite{Arman2015}. Moreover, a modification that can be made to the analysis given in this study is related to the regional upper bounds derived in Section \ref{sec:abs-construction-ta}. It has been observed that these bounds are relatively large since perturbations can have a stabilizing effect on the error dynamics and as a result cause the enlargement of inter-sample times. Since these large upper bounds can possibly result in a timed automaton with a large number of transitions, using these timed automata may suffer from scalability issues for scheduling purposes. Therefore, one may assume to arbitrate an upper bound on the triggering mechanism to facilitate the scheduling process. Considering the case of multiple WNCS's, this type of assumption is closely related to periodic-ETC strategies, see e.g.\ \cite{heemels2013periodic}, and the assumed bound can be seen as the network heartbeat forcing ETC feedback loops updates.

\section{ACKNOWLEDGMENTS}

The authors gratefully appreciate the fruitful discussions with Dieky Adzkiya and Nikolaos Kekatos.

\appendix

\textbf{Lemma \ref{lem3}:} Assume the time interval $[0,\sigma]$ is divided to $l$ subintervals. This step is related to the reduction of conservatism in polytopic embedding. Let $t \in [0,\sigma]$ be an instant such that it satisfies $j\frac{\sigma}{l}\leq t < (j+1)\frac{\sigma}{l}$ where $j \in \{0,\dots,l-1 \}$ and $t=t'+j\frac{\sigma}{l}$ ($t' \in [0,\chi]$, with $\chi =\frac{\sigma}{l}$ for $j< \lfloor \frac{\ubar{\tau}_s l}{\sigma} \rfloor$ and $\chi=\ubar{\tau}_s-j\frac{\sigma}{l}$ otherwise). Denote $\Lambda(t)-I$ by $\mathcal{X}(t)$. One has:
\begin{equation}
\label{eq_18}
\begin{split}
\mathcal{X}(t)=&\left[\int_0^{j\frac{\sigma}{l}} e^{As}ds+\right.\\
&\left.\int_0^{t'}e^{As}ds(A\int_0^{j\frac{\sigma}{l}}e^{As}ds+I)\right](A-BK).
\end{split}
\end{equation}

Rewrite (\ref{eq_18}) into a more compact form as follows:
\begin{equation}
\label{eq_19}
\mathcal{X}(t)=\check{\Pi}_j+\int_0^{t'}e^{As}ds \hat{\Pi}_j
\end{equation}
where
\begin{equation}
\label{eq_20}
\begin{array}{ll}
\check{\Pi}_j=\check{F}_j(A-BK),&\hat{\Pi}_j=\hat{F}_j(A-BK)\\
\check{F}_j=\int_0^{j\frac{\sigma}{l}} e^{As}ds,&\hat{F}_j=A\check{F}_j+I.
\end{array}
\end{equation}

Substitute (\ref{eq_19}) into (\ref{eq_17_123}), it follows:
\begin{equation}
\label{eq_22_1}
\begin{array}{ll}
\ubar{\Phi}_{11}(t)&=\check{\Pi}^T_j M \check{\Pi}_j+\check{\Pi}^T_j M (\int_0^{t'}e^{As}ds) \hat{\Pi}_j\\
&+\hat{\Pi}_j^T(\int_0^{t'}e^{As}ds)^T M \check{\Pi}_j\\
&+\hat{\Pi}_j^T(\int_0^{t'}e^{As}ds)^T M (\int_0^{t'}e^{As}ds) \hat{\Pi}_j\\
&+tW\mu \lambda_{\max}(E^TE)d_A(t)I- N,\\
\ubar{\Phi}_{12}(t)&=\check{\Pi}^T_j M+\hat{\Pi}_j^T(\int_0^{t'}e^{As}ds)^T M,\\
\ubar{\Phi}_{21}(t)&=\ubar{\Phi}_{12}^T(t),\\
\ubar{\Phi}_{22}(t)&=-\Psi.
\end{array}
\end{equation}

Now, we use the polytopic embedding approach proposed by \cite{hetel2006stabilization} to abstract away $t$ in (\ref{eq_17_123}). In the polytopic embedding approach, the underlying idea is as follows. First, we approximate the matrix functionals $tW\mu \lambda_{\max}(E^TE)d_A(t)I$ and $\ubar{\Phi}$ by their $N_{\text{conv}}$-th order Taylor series expansions. Note that one has:
\begin{equation}
\label{eq_23}
\int_0^{t'}e^{As}ds=\Sigma_{i=1}^{N_{\text{conv}}} \frac{A^{i-1}}{i!} (t')^i.
\end{equation}

Followed by these approximations, we take into account the introduced error and call the upper bound on this error $\eta$. The procedure to find $\eta$ follows. The exact Taylor expansion of $\Phi(t)$ is given by $\Sigma_{k=0}^{\infty} \hat{\ubar{\Phi}}_{k,j}(t')^k$ where $\hat{\ubar{\Phi}}_{k,j}$ is given in (\ref{eq_24_1}). However, in our analysis the $N_{\text{conv}}$-th order expansion of $\Phi(t)$, that is $\tilde{\ubar{\Phi}}_{(N_{\text{conv}},j)}(t')=\Sigma_{k=0}^{N_{\text{conv}}} \hat{\ubar{\Phi}}_{k,j}(t')^k$ has been used. Now, assume the error introduced by the approximation is $\ubar{R}_{(N_{\text{conv}},j)}(t')=\Phi(t)-\tilde{\ubar{\Phi}}_{(N_{\text{conv}},j)}(t')$ which happens to be a symmetric matrix. Hence, one is able to derive an upper bound on $\ubar{R}_{(N_{\text{conv}},j)}(t')\leq \eta I$ where the scalar $\eta$ is given by (\ref{eq_33}). It follows that $\tilde{\ubar{\Phi}}_{(N_{\text{conv}},j)}(t')+ \eta I \preceq 0$ implies $\Phi(t)\preceq 0$. Based on the fact that $\tilde{\ubar{\Phi}}_{(N_{\text{conv}},j)}(\cdot)+ \eta I$ is a polynomial function, one is able to use the convex embedding technique in \cite{hetel2006stabilization} to show that $\ubar{\Phi}_{(i,j),s} \preceq 0,\; \forall (i,j) \in \mathcal{K}_s=(\{0,\dots,N_{\text{conv}} \} \times \{0,\dots,\lfloor \frac{\ubar{\tau}_s l}{\sigma} \rfloor \})$, with $\ubar{\Phi}_{(i,j),s} = \sum_{k=0}^i L_{k,j} \chi^{k}+\eta I$ implies $ (\ubar{\tilde{\Phi}}_{N_{\text{conv}},j}(\sigma ')+\eta I)\preceq 0$ and as a result $\Phi(t) \preceq 0,\; \forall t \in [0,\ubar{\tau}_s]$.

\textbf{Theorem \ref{them3}:} Consider scalars $\ubar{\alpha}_{(i,j),s}$ for $n=2$ (or matrices $\ubar{U}_{(i,j),s}$ for $n\geq3$) satisfying LMI conditions given in (\ref{eq_them_un}) for $s\in \{1,\dots,q\}$. By the virtue of Schur complement and Lemma \ref{lem3}, it follows that $\Phi(t)+\ubar{\alpha}_{(i,j),s}Q_s\preceq0$ for $n=2$ (or $\Phi(t)+E_s^T\ubar{U}_{(i,j),s}E_s\preceq0$ for $n\geq3$). Then, since $\forall x\in \mathcal{R}_s$, $\{x\in \mathbb{R}^2|\;x^T Q_s x\geq0 \}$ for $n=2$ (or $\{x\in \mathbb{R}^n|\; E_s x\geq0 \}$ for $n\geq3$), the S-procedure implies that $x^T \Phi(t)x\leq 0$, $\forall t\in [0,\ubar{\tau}_s]$. Finally, Theorem \ref{theor_1} guarantees that $\forall x\in \mathcal{R}_s$, the inter-sample time $\tau(x)$ is lower bounded by $\ubar{\tau}_s$.

\textbf{Lemma \ref{lemm_u1}} and \textbf{Theorem \ref{them4}:} A sketch of proof is given. The polytopic embedding according to time of $-\Phi(t)$ enable us to show that $-\Phi(t)\preceq0$ (or $\Phi(t)\succeq0$) if $\bar{\Phi}_{\kappa,s} \preceq 0$, $\forall \kappa \in \mathcal{K}_s$. Then, by applying the Schur complement on $-\Phi(t)$, it follows $-\Theta(t)\preceq 0$ (or $\Theta(t)\succeq 0$) and as a result $-x^T \Theta(t)x\leq 0$. Furthermore, considering (\ref{eq_15}) in Theorem \ref{theor_1}, i.e.\ $-\mathcal{F}_{\omega}(x,t)\geq -x^T \Theta(t)x$, the claims in Lemma~\ref{lemm_u1} and Theorem~\ref{them4} follow.

\bibliographystyle{IEEEtran}
\bibliography{myjabref_comp}

\end{document}